\documentclass[preprint,12pt]{elsarticle}

\usepackage{amsmath,amssymb,amsbsy,amstext,amscd,amsfonts}
\usepackage{graphicx, color}
\usepackage{geometry}
\usepackage{subfigure}

\usepackage{float}

\newcommand{\bR}{\mathbb{R}}
\newcommand{\bC}{\mathbb{C}}
\newcommand{\bZ}{\mathbb{Z}}
\newcommand{\bN}{\mathbb{N}}

\newcommand{\cV}{\mathcal{V}}
\newcommand{\cG}{\mathcal{G}}

\newcommand{\Rint}{\mathring{R}}

\newtheorem{theorem}{\noindent {\rm \bf Theorem}}[section]

\newtheorem{definition}[theorem]{\noindent {\rm \bf Definition}}

\newenvironment{proof}{\begin{trivlist}\item[]{\sl Proof.\/\ }}
                      {\hfill $\Box$
                      \end{trivlist}}

\journal{arXiv.org}

\begin{document}
\begin{frontmatter}

\title{Half-plane diffraction problems on a triangular lattice}

\author[rmi,fre]{D. Kapanadze \corref{cor1}}
\ead{david.kapanadze@gmail.com}

\author[rmi]{E. Pesetskaya }
\ead{kate.pesetskaya@gmail.com}

\cortext[cor1]{Corresponding author}
\address[rmi]{A. Razmadze Mathematical Institute, TSU, Merab Aleksidze II Lane 2, Tbilisi 0193, Georgia}
\address[fre]{Free University of Tbilisi, Tbilisi 0159, Georgia}

\begin{abstract}
We investigate thin-slit diffraction problems for two-dimensional lattice waves. The peculiar structure allows us to consider the problems on the semi-infinite triangular lattice, consequently, we study Dirichlet problems for the two-dimen-sional discrete Helmholtz equation in a half-plane. In view of the existence and uniqueness of the solution, we provide new results for the real wave number $k\in (0,3)\backslash\{2\sqrt{2}\}$ without passing to the complex wave number and derive an exact representation formula for the solution. For this purpose, we use the notion of the radiating solution. Finally, we propose a method for numerical calculation. The efficiency of our approach is demonstrated in an example related to the propagation of wave fronts in metamaterials through two small openings.
\end{abstract}

\begin{keyword}
discrete Helmholtz equation, Dirichlet boundary value problem, half-plane diffraction, metamaterials, triangular lattice model
\end{keyword}

\end{frontmatter}

\section{Introduction}

Wave propagation through discrete structures remains an active area of research today.  The triangular lattice is one of the five two-dimensional Bravais lattice types and appears naturally in applications. For example, close-packed planes  occur frequently in some kinds of crystals in the form of triangular lattice \cite{BH54, Bu66}. Besides, one can consider two-dimensional passive propagation media, a host microstrip line network periodically loaded with series capacitors and shunt inductors for signal processing and filtering (as shown in Figure \ref{Fig_1}). This type of inductor-capacitor lattice is referred to a negative-refractive-index transmission-line (NRI-TL) metamaterial \cite{CI} or, simply, left-handed 2D metamaterial.

Motivated by applications of recent interest related to analog circuits, crystalline materials, and metamaterials, we investigate thin-slit diffraction problems for triangular lattices. The peculiar structure allows us to consider the problems on the semi-infinite triangular lattice. For this lattice, we study Dirichlet problems for the two-dimensional discrete Helmholtz equation in a half-plane (mathematically formulated in the next section). When one is interested in the analysis of regular processes in which waves corresponding to the microstructural scales can be neglected, then the continuum limit of corresponding equations can be investigated. In this case, we arrive at the famous problem in applied mathematics – the half-plane diffraction problem. It is well known that the classical continuum model of wave diffraction can be considered only as the slowly-varying approximation of a discrete or structured material. Nowadays the understanding of nanostructure and microstructure phenomena in modern materials and composites in critical conditions is of crucial importance. It turns out that a discrete structure provides an effective way to describe microstructural processes and influences within, cf., e.g., \cite{Br, CI, Do, Sl}. Therefore, there is an obvious necessity to avoid continuum limits and instead directly analyze the discrete Helmholtz diffraction problems. Although the similar problem for square lattice has been studied in \cite{BO}, see also \cite{DK3}, its extension to a triangular lattice model is not direct.

Our main interest lies in the investigation of the problems with wave number $k$ within the pass-band which is an arbitrary non-zero complex number in general. In this paper, we use the well approved approach to provide new results for a real wave number $k\in (0,3)\backslash\{2\sqrt{2}\}$ without passing to a complex wave number.
In order to have the unique solvability results and an exact representation formula for the solution, we use the notion of the radiating solution \cite{DK2} and the method of images to construct the Dirichlet Green’s function for the half-plane. Notice that the cases with the wave numbers within the stop-band is mathematically more simple, they do not need the radiation conditions and is not considered here. Finally, we propose a method for numerical calculation. The efficiency of our approach is demonstrated in an example related to the propagation of wave fronts in metamaterials through two small openings.

\section{Basic notations and formulation of the problem}

Following the customary notation in mathematics, let $\bZ$, $\bZ^+$, $\bZ^-$, $\bN$, $\bR$, and $\bC$ denote the set of integers, positive integers, negative integers, non-negative integers, real numbers and complex numbers, respectively. We denote by $e_1=(1,0)$, $e_2=(0,1)$ the standard base of $\bZ^2$ ($=\bZ\times \bZ$).

Consider a two-dimensional infinite triangular lattice $\mathfrak{T}$ defined as a periodic simple graph $\{\mathcal{V}, \mathcal{E}\}$, where
\[
\mathcal{V}=\{(x_1+x_2/2,\sqrt{3}x_2/2)\subset \mathbb{R}^2 : (x_1,x_2)\in\mathbb{Z}^2\}
\]
is a vertex set, and $\mathcal{E}$ is an edge set, whose endpoints $(a,b)\in \mathcal{V}\times \mathcal{V}$ are adjacent points, i.e., $\mid a-b\mid=1$, cf. Figure \ref{Fig_1}. The time-harmonic discrete waves in $\mathfrak{T}$ can be described by solutions of the discrete Helmholtz equation

\begin{figure}[t!]  
\subfigure{\includegraphics[scale=0.42]{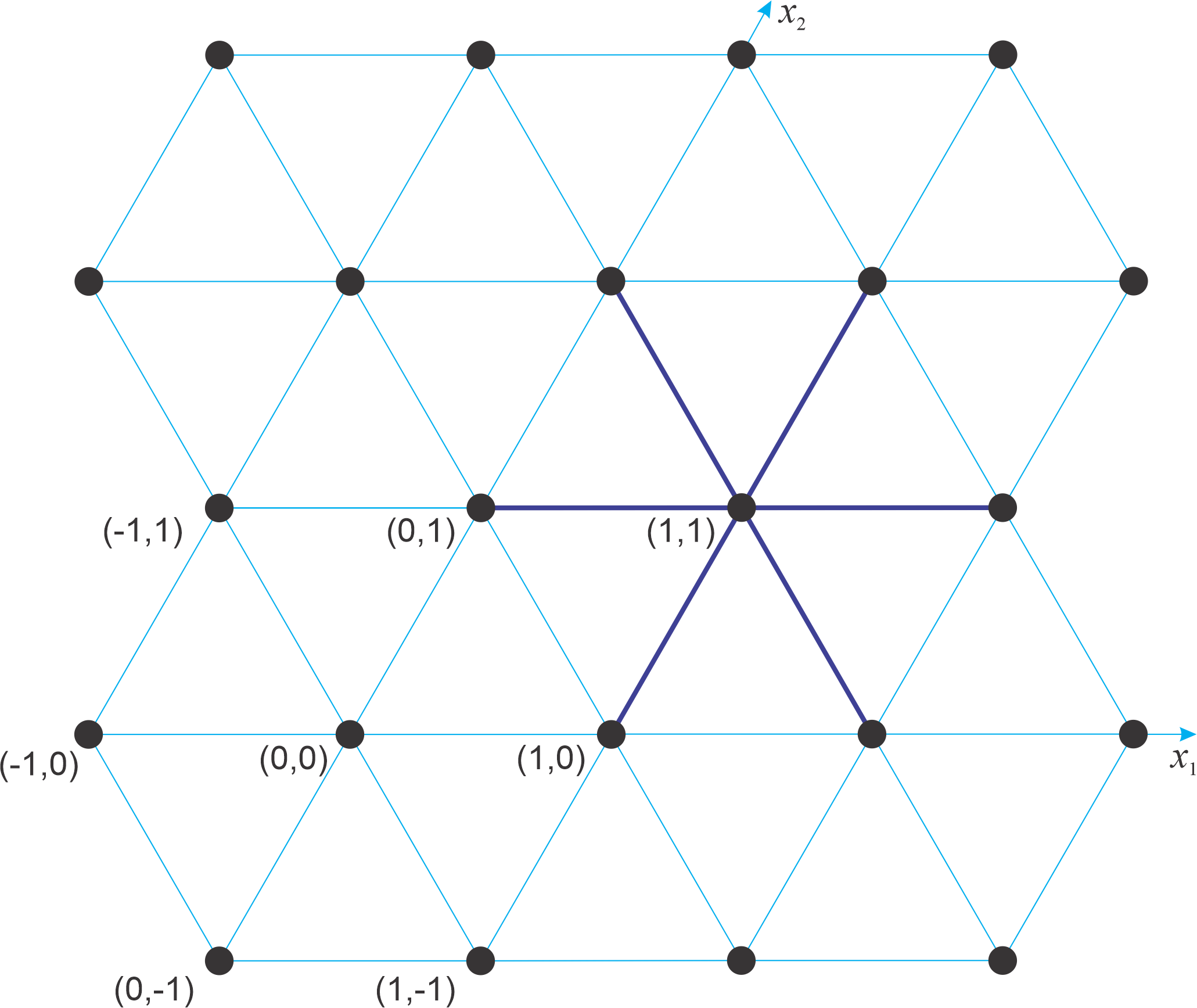}}
\subfigure{\raisebox{18mm}{\includegraphics[scale=0.50]{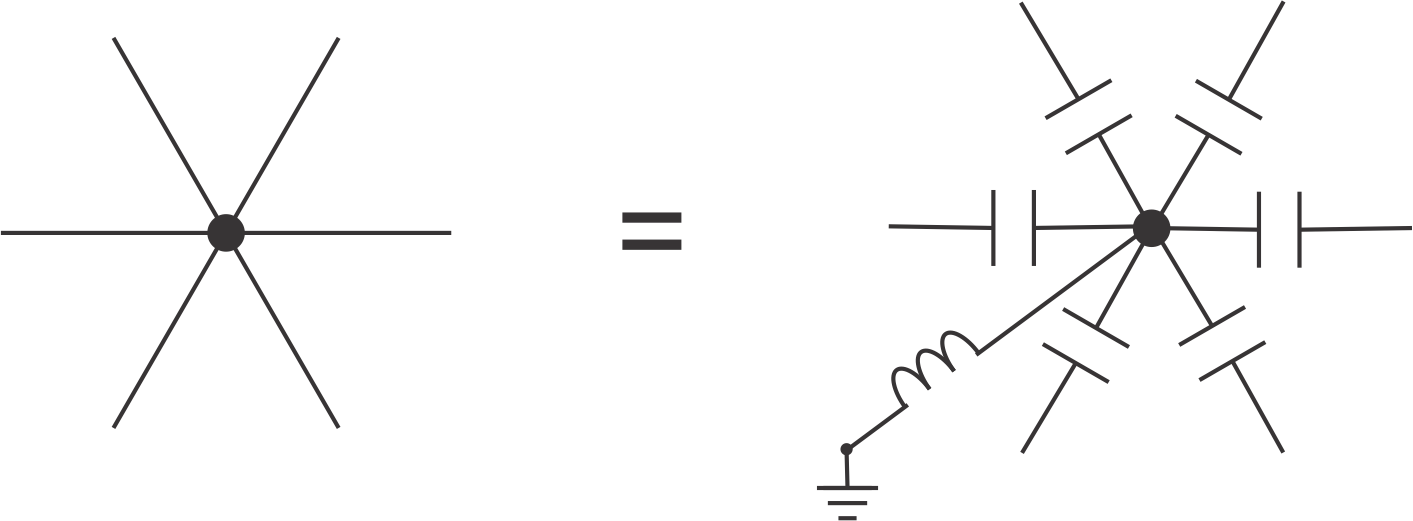}}}
\caption{Triangular lattice. Connection between $x=(x_1,x_2)\in \bZ^2$ and the Euclidean coordinates of the vertexes (black dots) is established via  $(x_1,x_2) \to (x_1+x_2/2,\sqrt{3}x_2/2)$. The nearest neighbor interactions based on the triangular lattice are shown with thick blue lines. Triangular lattice can be viewed as a left-handed 2D inductor-capacitor metamaterial. A host transmission-line is loaded periodically with series capacitors and shunt inductors.}
\label{Fig_1}
\end{figure}

\begin{equation}\label{eq:Helmholtz}
(\Delta_d + k^2)u(x)=g(x), \quad x=(x_1,x_2)\in \bZ^2.
\end{equation}
Here, $\Delta_d$ denotes the discrete (a 7-point) Laplacian defined as follows
\begin{equation*}
\begin{aligned}
\Delta_d u(x)=&u(x+e_1)+u(x-e_1)+u(x+e_2)+u(x-e_2)\\
&+u(x+e_1-e_2)+u(x-e_1+e_2)-6u(x),
\end{aligned}
\end{equation*}
where $g$ is a function with bounded support, i.e., $g\in C_0(\bZ^2)$.

 Let $\Omega=\mathring{\Omega}\cup \Gamma$ be a discrete half-plane in $\bZ^2$, where
$\mathring{\Omega}:=\mathbb{Z}\times \mathbb{\bZ^+}$ and $\Gamma := \partial \Omega =\{(x_1,0) : x_1\in \mathbb{Z}\}$.
We consider the Dirichlet poblem for the discrete Helmholtz equation on the semi-infinite lattice:

\begin{subequations}
\begin{align}
(\Delta_d +k^2)u(x) &= 0,\quad\quad\  \ \textup{in}\  \mathring{\Omega}, \label{eq:H1} \\
u(y)&= f(y),\quad \textup{on}\  \Gamma. \label{eq:H2}
\end{align}
\end{subequations}
Here, $f$ is a given function supported on a finite subset of $\Gamma$.

Thus, we are interested in studying the problem of the existence and uniqueness of a function $u:\Omega\to \bC$ such that $u$ satisfies the discrete Helmoltz equation \eqref{eq:H1} with $k \in (0,3)\backslash \{2\sqrt{2}\}$ and the boundary condition \eqref{eq:H2}. From now on we will refer to this problem as Problem $\mathcal{P}_{\mathrm{D}}$.

\section{Lattice Green's function}

Denote by $\cG(x;y)$ the Green's function for the discrete Helmholtz equation \eqref{eq:Helmholtz} centered at $y$ and evaluated at $x$. Then the function $\cG(x;y)$ satisfies
\begin{equation}\label{eq:greens}
(\Delta_d+k^2)\cG(x;y)=\delta_{x,y},
\end{equation}
where $\delta_{x,y}$ is the Kronecker delta. For brevity, we use the notation $\cG(x)$ for $\cG(x;0)$. Notice that
$\cG(x;y)=\cG(x-y)$.

Using the discrete Fourier transform and the inverse Fourier transform
we get
\begin{equation}\label{eq:gmn}
\cG(x)=\frac{1}{4\pi^2}\int_{-\pi}^{\pi}\int_{-\pi}^{\pi}\frac{e^{\iota(x\cdot \xi)}}{\sigma(\xi;k)}d\xi, \quad \xi=(\xi_1,\xi_2),
\end{equation}
where
\begin{equation}
\begin{aligned}
\sigma(\xi;k^2)&=e^{\iota\xi_1}+e^{-\iota\xi_1}+e^{\iota\xi_2}+e^{-\iota\xi_2}+e^{\iota\xi_1}e^{-\iota\xi_2}+e^{-\iota\xi_1}e^{\iota\xi_2}-6+k^2\\
&=k^2-6+2\cos\xi_1+2\cos\xi_2+2\cos(\xi_1-\xi_2).
\end{aligned}
\end{equation}

The lattice Green's function $\mathcal{G}$ is quite well known when $k^2\in \bC\backslash[0, 9]$ (cf., e.g., \cite{Ho}).
Notice that if $k^2\in \bC\backslash[0, 9]$ then $\sigma\neq 0$ and, consequently, $\mathcal{G}$ in \eqref{eq:gmn} is well defined. In this case
$\mathcal{G}(x)$ decays exponentially when $\mid x \mid\to\infty$. Moreover, we have
\begin{equation}\label{eq:lateq}
\mathcal{G}(x_1,x_2)=\mathcal{G}(x_2,x_1)=\mathcal{G}(-x_1,-x_2)=\mathcal{G}(x_1+x_2,-x_2)
\end{equation}
for all $x=(x_1,x_2) \in \bZ^2$.

Let us show $\mathcal{G}(x_1,x_2)=\mathcal{G}(x_1+x_2,-x_2)$ while other identities are trivial.
According to \eqref{eq:gmn}, we have
\begin{eqnarray*}
\mathcal{G}(x_1+x_2,-x_2) & = & \frac{1}{4\pi^2}\int_{-\pi}^{\pi}\int_{-\pi}^{\pi}\frac{e^{\iota x_1\xi_1 +\iota x_2(\xi_1-\xi_2)}}{\sigma(\xi_1,\xi_2;k)}d\xi_1 d\xi_2 \\
& = &  \frac{1}{4\pi^2}\int_{-\pi}^{\pi}\int_{\eta_1-\pi}^{\eta_1+\pi}\frac{e^{\iota x_1\eta_1 +\iota x_2\eta_2}}{\sigma(\eta_1,\eta_2;k)}d\eta_2 d\eta_1
\end{eqnarray*}
with $\eta_1=\xi_1, \, \eta_2=\xi_1-\xi_2$.
The following equality can be easily obtained by changing the variable $\eta_2$ to $\eta_2-2\pi$
\[
\int_{0}^{\pi}\int_{\pi}^{\eta_1+\pi}\frac{e^{\iota x_1\eta_1 +\iota x_2\eta_2}}{\sigma(\eta_1,\eta_2;k)}d\eta_2 d\eta_1=
\int_{0}^{\pi}\int_{-\pi}^{\eta_1-\pi}\frac{e^{\iota x_1\eta_1 +\iota x_2\eta_2}}{\sigma(\eta_1,\eta_2;k)}d\eta_2 d\eta_1.
\]
 The factor $e^{\iota x_2 2\pi}$ is equal to $1$ since $x_2\in \mathbb{Z}$. Similar arguments give us
\[
\int_{-\pi}^{0}\int_{\eta_1-\pi}^{-\pi}\frac{e^{\iota x_1\eta_1 +\iota x_2\eta_2}}{\sigma(\eta_1,\eta_2;k)}d\eta_2 d\eta_1=
\int_{-\pi}^{0}\int_{\eta_1+\pi}^{\pi}\frac{e^{\iota x_1\eta_1 +\iota x_2\eta_2}}{\sigma(\eta_1,\eta_2;k)}d\eta_2 d\eta_1.
\]
Taking into account the last two equalities, we get
\[
\mathcal{G}(x_1+x_2,-x_2) = \frac{1}{4\pi^2}\int_{-\pi}^{\pi}\int_{-\pi}^{\pi}\frac{e^{\iota x_1\eta_1 +\iota x_2\eta_2}}{\sigma(\eta_1,\eta_2;k)}d\eta_2 d\eta_1 = \mathcal{G}(x_1,x_2).
\]

In the case $k^2\in (0, 9)\backslash\{8\}$ the expression \eqref{eq:gmn} is understood as follows: we replace $k^2$ by $k^2+\iota\varepsilon$ with $0<\varepsilon\ll 1$ and let $\varepsilon\to 0$ at the end of the calculation, cf. \cite{DK2}. Thus, we define the lattice Green's function for $k\in (0, 3)\backslash\{2\sqrt{2}\}$ as a pointwise limit of
\[
(R_{\lambda+\iota\varepsilon}\delta_{x,0})(x):=\frac{1}{4\pi^2}\int_{-\pi}^{\pi}\int_{-\pi}^{\pi}\frac{e^{\iota x\cdot\xi }d\xi_1 d\xi_2}{\sigma(\xi;k^2+\iota\varepsilon)}
\]
as $k^2+\iota\varepsilon \to k^2+\iota 0$ and denote it again by $\mathcal{G}(x)$, i.e., $\mathcal{G}(x)=(R_{\lambda+\iota 0}\delta_{x,0})(x) $. Clearly, $\mathcal{G}(x)$ is a solution to equation \eqref{eq:greens} and satisfies equalities  \eqref{eq:lateq}.

\section{Unique solvability result}

In order to simplify further arguments let us introduce following vectors:
\begin{equation*}
\begin{aligned}
e_1=(1,0),\quad e_2=(0,1),\quad e_3=e_1-e_2,\quad e_4=-e_1,\quad e_5=-e_2,\quad e_6=-e_3.
\end{aligned}
\end{equation*}
For any point $x\in\bZ^2$ we define the 6-neighbourhood $F^0_x$ as the set of points $\{x+e_j: j=1,...,6\}$ and the neighbourhood $F_x$ as $F^0_x\bigcup\{x\}$. We say that $R\subset \bZ^2$ is a region if there exist disjoint nonempty subsets $\Rint$ and $\partial R$ of $R$ such that
\begin{itemize}
\item[(a)] $R=\Rint\cup \partial R$,
\item[(b)] if $x\in \mathring{R}$ then $F_x\subset R$,
\item[(c)] if $x\in\partial R$ then there is at least one point $y\in F^0_x$ such that $y\in \Rint$.
\end{itemize}
Clearly, the subsets $\Rint$ and $\partial R$ are not defined uniquely by $R$, but henceforth, it will be always assumed that $\Rint$ and $\partial R$ are also given and fixed for a given region $R$ in $\bZ^2$. We also say that $x$ is an interior (boundary) point of $R$ if $x\in \Rint$ ($x\in \partial R$).

Denote by $(\partial R)_j$, $j=1,...,6$, a set of all boundary points $y\in\partial R$ such that $y-e_j\in \Rint$ and call it the sides of the boundary $R$. Note that a boundary point $x$ can simultaneously belong to all six sides of $R$. Thus, these sides may overlap each other. Clearly, $\partial R$ is the union of its six sides, $\partial R=\cup_{j=1}^6 (\partial R)_j$.

As it was already mentioned above, $y\in\partial R$ may be a point of intersection of several sides of $\partial R$. However, in our arguments presented below it will be always clear which side is needed to be considered. Under this condition, we define the discrete derivative in the outward normal direction $e_j$, $j=1,\dots,6$,
\[
Tu(y)=u(y)-u(y-e_j), \quad y\in(\partial R)_j.
\]

Let us introduce the following set $H_0=\{(0,0)\}$ and then define $H_{N}$, $N\in\mathbb{N}$, with the help of recurrence formula
\[
H_{N}:=\bigcup_{x\in H_{N-1}}F_{x}
\]
with $\mathring{H}_N:=H_{N-1}$ and $(\partial H)_{N}:=H_{N}\backslash \mathring{H}_N$.

Let $R$ be a finite region. Representing $R=\bigcup_{x\in \Rint}{F_x}$, we can easily derive a discrete analogue of Green's second identity
\begin{equation}\label{eq:GRsecond}
\sum_{x\in \Rint}(u(x)\Delta_d v(x)-v(x)\Delta_d u(x))=\sum_{y\in \partial R} (u(y)Tv(y)-v(y)Tu(y)).
\end{equation}

Now let us give a definition of a radiating solution on the discrete half-plane. First, we consider the case $k^2\in (0,8)$.
We say that $u:\Omega \to\bC$ satisfies the radiation condition at infinity if
\begin{equation}
\label{eq:radcond}
\left\{\begin{aligned} u(x)&=O(\mid x\mid^{-\frac12}),\\
u(x+e_j)&=e^{\iota \xi^*_j(\alpha,k)} u(x)+O(\mid x\mid^{-\frac32}),\quad j=1,2,
\end{aligned}
\right.
\end{equation}
with the remaining term decaying uniformly in  all directions $x/\mid x\mid$, where $x$ is characterized as $x_1=\mid x\mid\cos\alpha$,  $x_2=\mid x\mid\sin\alpha$, $0\le \alpha \le \pi$. Here, $\xi^*_j(\alpha,k)$ is the $j$th coordinate of the point $\xi^*(\alpha, k)$.

For $k^2\in (8,9)$ we say that $u:\Omega\to\bC$ satisfies the radiation condition at infinity if it can be represented as a sum of functions $u_1$ and $u_2$ such that each of them satisfies
\begin{equation}
\label{eq:radcond2}
\left\{\begin{aligned} u_i(x)&=O(\mid x\mid^{-\frac12}),\\
u_i(x+e_j)&=e^{\iota \xi^{*,i}_j(\alpha,k)} u_i(x)+O(\mid x\mid^{-\frac32}),\quad i,j=1,2,
\end{aligned}
\right.
\end{equation}
with the remaining term decaying uniformly in  all directions $x/\mid x\mid$, where $x$ is characterized as $x_1=\mid x\mid\cos\alpha$,  $x_2=\mid x\mid\sin\alpha$, $0\le \alpha \le \pi$. Here, $\xi^{*,i}_j(\alpha,k)$ is the $j$th coordinate of the point $\xi^{*,i}(\alpha, k)$, $i=1,2$.

It is shown in \cite{DK2} that $\mathcal{G}(x)$ satisfies the radiation conditions introduced above. Now we are ready to introduce a notion of a radiating solution to the discrete Helmholtz equation \eqref{eq:Helmholtz}.

\begin{definition} Let $k^2\in (0,8)$. A solution $u$ to the discrete Helmholtz equation \eqref{eq:H1} is called radiating if it satisfies the radiation condition \eqref{eq:radcond}.

Let $k^2\in (8,9)$. A solution $u$ to the discrete Helmholtz equation \eqref{eq:H1} is called radiating if it satisfies the radiation condition \eqref{eq:radcond2}.
\end{definition}

We represent the Dirichlet Green's function for the half-plane as follows:
\[
\mathcal{G}^+(x;y)=\mathcal{G}(x;y)-\mathcal{G}(\hat{x};y), \quad y\in \bZ^2,
\]
with $x=(x_1,x_2)\in \Omega$ and $\hat{x}=(x_1+x_2,-x_2)$. Notice that  $\mathcal{G}^+(x;y)$ can be represented equivalently as
\[
\mathcal{G}^+(x;y)=\mathcal{G}(x;y)-\mathcal{G}(x;\hat{y}).
\]
Indeed, using the property \eqref{eq:lateq}, we have
\begin{equation}\label{eq:Greenyhat}
\mathcal{G}(\hat{x};y)=\mathcal{G}(x_1+x_2-y_1,-x_2-y_2)=\mathcal{G}(x_1-y_1-y_2, x_2+y_2)=\mathcal{G}(x; \hat{y}).
\end{equation}
Since $\widehat{(y_1,0)}=(y_1,0)$, we immediately get $\mathcal{G}^+(x;y)=0$ for $y \in \Gamma$. Moreover,
\[
(\Delta_d+k^2)\mathcal{G}^+(x;y)=(\Delta_d+k^2)\mathcal{G}(x;y)-(\Delta_d+k^2)\mathcal{G}(\hat{x};y)=\delta_{x,y}-\delta_{\hat{x},y}=\delta_{x,y}
\]
for all $x\in \mathring{\Omega}$ and $y\in \Omega$. From \eqref{eq:Greenyhat} we see that $\mathcal{G}^+(x;y)$ satisfies the radiation condition.  Indeed,
for any fixed $y\in \Gamma$ the angle $\alpha$ from the radiation conditions \eqref{eq:radcond} and \eqref{eq:radcond2} is the same for $\mathcal{G}(x;y)=\mathcal{G}(x-y)$ and $\mathcal{G}(x;\hat{y})=\mathcal{G}(x-\hat{y})$ when $\mid x\mid \to \infty$.

\begin{theorem}\label{th:repr}
Let $k \in (0,3)\backslash \{2\sqrt{2}\}$, and $u$ be a given function $\Omega \to \bC$ that satisfies the radiation condition. The function $u\mid_{\Gamma}$ has a finite support. Then, for any point $x\in \mathring{\Omega}$, we have the following representation formula
\[
u(x)=\sum_{y\in \Gamma}\big(u(y)T\cG^+(x;y)-\cG^+(x;y)Tu(y)\big)+\sum_{y\in \mathring{\Omega}}\cG^+(x;y)(\Delta_d+k^2) u(y).
\]
In particular, if $u$ is a solution to the discrete Helmholtz equation
\[
(\Delta_d+k^2)u(x)=0\quad \textup{in}\  \mathring{\Omega},
\]
then
\begin{equation}\label{eq:repr}
u(x)=\sum_{y\in \Gamma}u(y)T\cG^+(x;y).
\end{equation}
\end{theorem}

\begin{proof}
Denote by $\tilde{u}$ the extension of $u$ to $\bZ^2$ by zeros, i.e,
$\tilde{u}(x)=u(x)$ if $x\in \Omega$ and $\tilde{u}(x)=0$ if $x\in \bZ\times\bZ^-$. Then, for any finite region $H_N$, $N\in\mathbb{N}$, we apply Green's second identity \eqref{eq:GRsecond} where we take $v(y)=\cG^+(x;y)$. Here, note that $\cG^+$ is a fundamental solution for $y\in \mathring{\Omega}$, and
$\tilde{u}(y)=0$ for $y\in \bZ\times\bZ^-$. Thus, $u(y)(\Delta_d+k^2) v(y)$ disappears in the following identity
\[
u(y)\Delta_d v(y)-v(y)\Delta_d u(y)=u(y)(\Delta_d+k^2) v(y)-v(y)(\Delta_d+k^2) u(y)
\]
as far as $y\in \bZ\times\bZ^-$. For $x\in \mathring{H}_N$, we get
\begin{align*}
\tilde{u}(x)&=\sum_{y\in \mathring{H}_N}\tilde{u}(y)\delta_{x,y}\\
       &=\sum_{y\in \partial H_N} \big(\tilde{u}(y)T\cG^+(x;y)-\cG^+(x;y)T\tilde{u}(y)\big)+\sum_{y\in \mathring{H}_N}(\cG^+(x;y)(\Delta_d+k^2) \tilde{u}(y)).
\end{align*}

Passing to the limit $N\to \infty$, we use exactly the same arguments as in the proof of Theorem 5.2 from \cite{DK2} and, consequently, we get
\[
\tilde{u}(x) = \sum_{y\in \bZ^2}\cG^+(x;y)(\Delta_d+k^2) \tilde{u}(y).
\]
For the function $u$, we obtain
\[
u(x)=\sum_{y\in \Gamma}\big(u(y)T\cG^+(x;y)-\cG^+(x;y)Tu(y)\big)+\sum_{y\in \mathring{\Omega}}\cG^+(x;y)(\Delta_d+k^2) u(y).
\]
Taking into the account that $\cG^+(x;y)=0$ when $y\in \Gamma$, for a solution $u$ to the discrete Helmholtz equation we get the following quality
\[
u(x)=\sum_{y\in \Gamma}u(y)T\cG^+(x;y), \quad x\in \mathring{\Omega}.
\]
\end{proof}

Now we are ready to prove the unique solvability result for the discrete Helmholtz equation on the semi-infinite triangular lattice.
\begin{theorem}
Let $k \in (0,3)\backslash \{2\sqrt{2}\} $ then the Problem $\mathcal{P}_{\mathrm{D}}$ has a unique radiating solution $u$ which
can be represented as
\begin{equation}\label{eq:HPsolution}
u(x)=\sum_{y\in \Gamma}(\delta_{x,y}-\mathcal{G}^+(x;y+e_2)-\mathcal{G}^+(x,y+e_2-e_1))f(y).
\end{equation}
\end{theorem}
\begin{proof} To prove the uniqueness result, it is sufficient to show that the corresponding homogeneous problem has only
the trivial solution. Let $u$ be a radiating solution to the homogeneous problem $\mathcal{P}_D$. Then Theorem \ref{th:repr} immediately implies $u(x)=0$ for all $x\in \mathring{\Omega}$.

To show the existence results, let us first check that $u(x)$ satisfies the boundary condition \eqref{eq:H2}. Since $\mathcal{G}^+(x;y)=0$ for all $x\in \Gamma$ and any $y$ we get
\[
u(x_1,0)=\sum_{y_1\in \mathbb{R}}\delta_{x_1,y_1}\delta_{0,0}f(y)=f(x_1).
\]
Now let us check that $u(x)$ satisfies the Helmholtz equation \eqref{eq:H1}. For $x_2>1$ all terms $(\Delta_d+k^2)\delta_{x,y}$,  $(\Delta_d+k^2)\mathcal{G}^+(x;y+e_2)=\delta_{x,y+e_2}$ and $(\Delta_d+k^2)\mathcal{G}^+(x;y+e_2-e_1)=\delta_{x,y+e_2-e_1}$ are equal to zero. Consequently, $(\Delta_d+k^2)u(x)=0$  for points $x=(x_1,x_2)$ with $x_2>1$.
It remains to consider the case $x_2=1$. By the direct calculation we have
\begin{eqnarray*}
  (\Delta_d+k^2)\delta_{x,y}f(y)&=&f(x_1)+f(x_1+1),\\
  (\Delta_d+k^2)\mathcal{G}^+(x;y+e_2)f(y)&=&f(x_1),\\
  (\Delta_d+k^2)\mathcal{G}^+(x;y+e_2-e_1)f(y)&=&f(x_1+1),
\end{eqnarray*}
and as a result $(\Delta_d+k^2)u(x)=0$. Thus, $u(x)$ is a solution to \eqref{eq:H1}. Since $\mathcal{G}^+$ satisfies the radiation condition and $f$ is supported on the finite subset of $\Gamma$ it follows that $u(x)$ is the unique radiating solution to the problem \eqref{eq:H1}, \eqref{eq:H2}.
\end{proof}

\section{Numerical results}

The main issue for numerical evaluation of the solution \eqref{eq:HPsolution} is to compute the lattice Green's function. For this purpose we apply  the method developed in \cite{BC}. Using 8-fold symmetry, we need only to compute the lattice Green's function $\cG(i,j)$ with $i\ge j \ge 0$. Following to \cite{BC}, let us introduce the vectors
$\cV_{2p}=(\cG(2p,0),\cG(2p-1,1),\dots,\cG(p,p))^\top$ and $\cV_{2p+1}=(\cG(2p+1,0),\cG(2p,1),\dots,\cG(p+1,p))^\top$ that collect all distinct Green’s functions $\cG(i,j)$ with ``Manhattan distances" $\mid i \mid + \mid j \mid$ of $2p$ and $2p+1$, respectively. For any Manhattan distance larger than 1, equation
\[
(\Delta_d+k^2)\cG(x)=\delta_{x,0}
\]
can be written in the matrix form $\gamma_n(k)\cV_n=\alpha_n(k)\cV_{n-1}+\beta_n(k)\cV_{n+1}$
where $\alpha_n(k)$, $\beta_n(k)$ and $\gamma_n(k)$ are sparse matrices (cf., Appendix A). Notice that only the
dimensions of these matrices depend on $n$. It is shown in \cite{BC} that, for any $n\ge 1$, we have
\[
\cV_n=A_n(k)\cV_{n-1},
\]
where the matrices $A_n(k)$ are defined by the following recurrence formula
\[
A_n(k)=[\gamma_n(k)-\beta_n(k)A_{n+1}]^{-1}\alpha_n(k).
\]
They can be computed starting from a sufficiently large $N$ with $A_{N+1}(k)=0$. Here, it is worth mentioning, that for $k=2$ we need to choose a better ``initial guess" than $A_{N+1}(k)=0$, since in this case $\mathrm{det}\gamma_n(k)=0$ and the matrix $\gamma_n(k)-\beta_n(k)A_{n+1}$ is not invertible.

Once $A_n(k)$ are known, we have $\cV_n=A_n(k)\dots A_1(k)\cV_{0}$, where
$\cV_{0}=\cG(0,0)$. In particular, $\cV_1=\cG(1,0)=A_1(k)\cG(0,0)$ which, together with
$6\cG(1,0)-(6-k^2)\cG(0,0)=1$, gives $\cG(0,0)=1/[6A_1(k)-6+k^2]$. This completes the calculation of the Green's function using elementary operations and no integrals. Notice also one more important advantage of this method. The $A_n(k)$ matrices are calculated coming down from asymptotically large Manhattan distances.  As they are propagated towards smaller Manhattan distances, it definitely gives us the physical solution.

Finally, we demonstrate our theoretical and numerical approaches for the diffraction problem on triangular lattice. For this purpose we take $k=\sqrt{2}$ and consider the wave $U(x_1,x_2,t)=\exp(\iota x_2 \pi/3-\iota\omega t)$ on the semi-infinite lattice $\mathbb{Z}\times \mathbb{Z}^-$ which encounters an obstacle at $x_2=0$ with two small openings formed by four nodes $\{(-11,0),(-10,0), (10,0), (11,0)\}$. It is easy to check that $\exp(\iota x_2 \pi/3)$ satisfies the discrete Helmholtz equation. Thus, our goal is to evaluate numerically a radiating solution of Problem $\mathcal{P}_{\mathrm{D}}$ where, for the given data, we take $f(\pm 10)=f(\pm 11)=1$ and  $f(m)=0$, $m\in\bZ^+\backslash \{-11,-10,10,11\}$.

\begin{figure}[t!]  
\subfigure[The density plot of $\mathrm{Re}\,\mathcal{G}$.]{\includegraphics[scale=0.45]{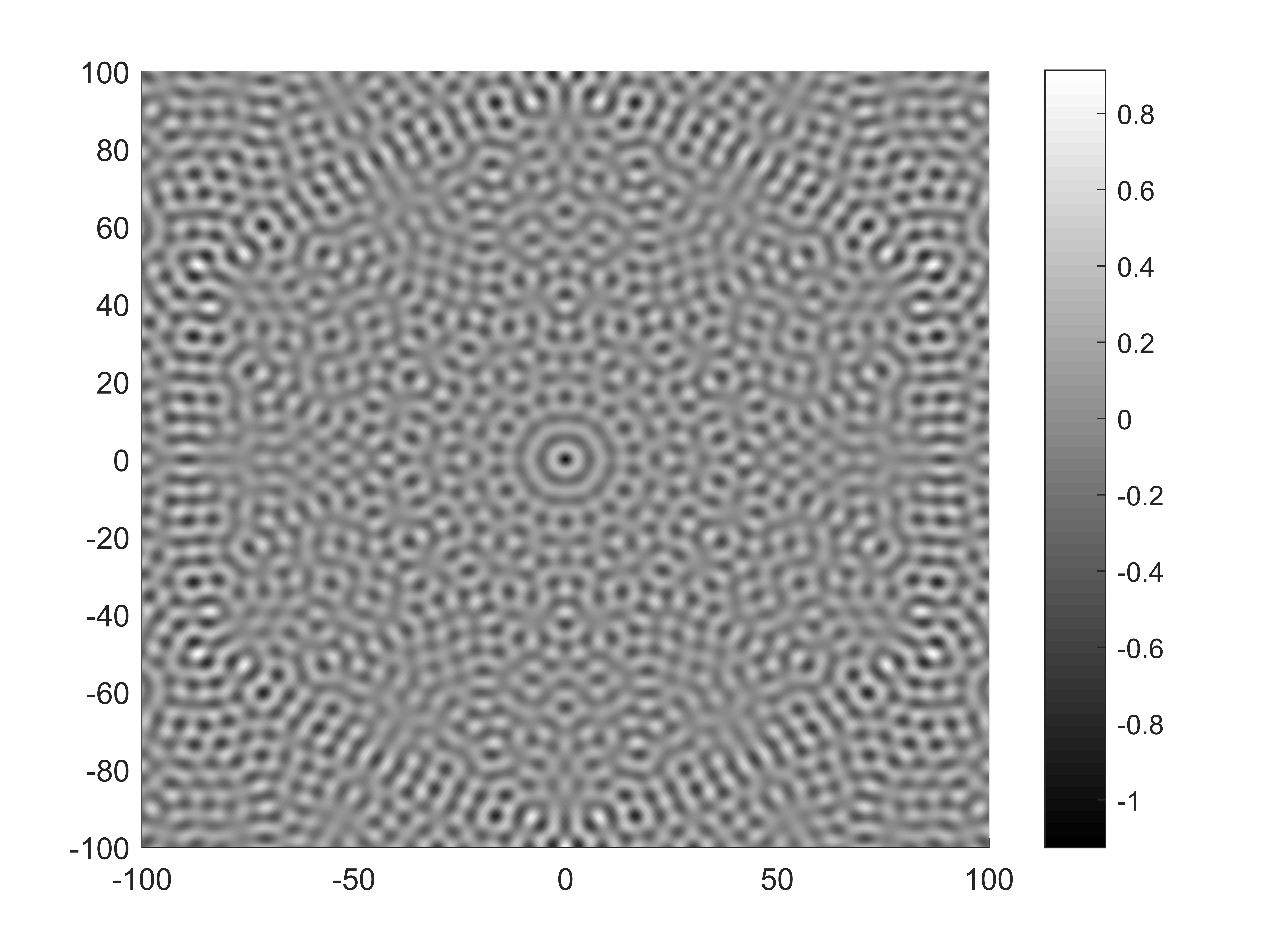}}
\subfigure[The density plot of $\mathrm{Re}\,u$.]{\includegraphics[scale=0.45]{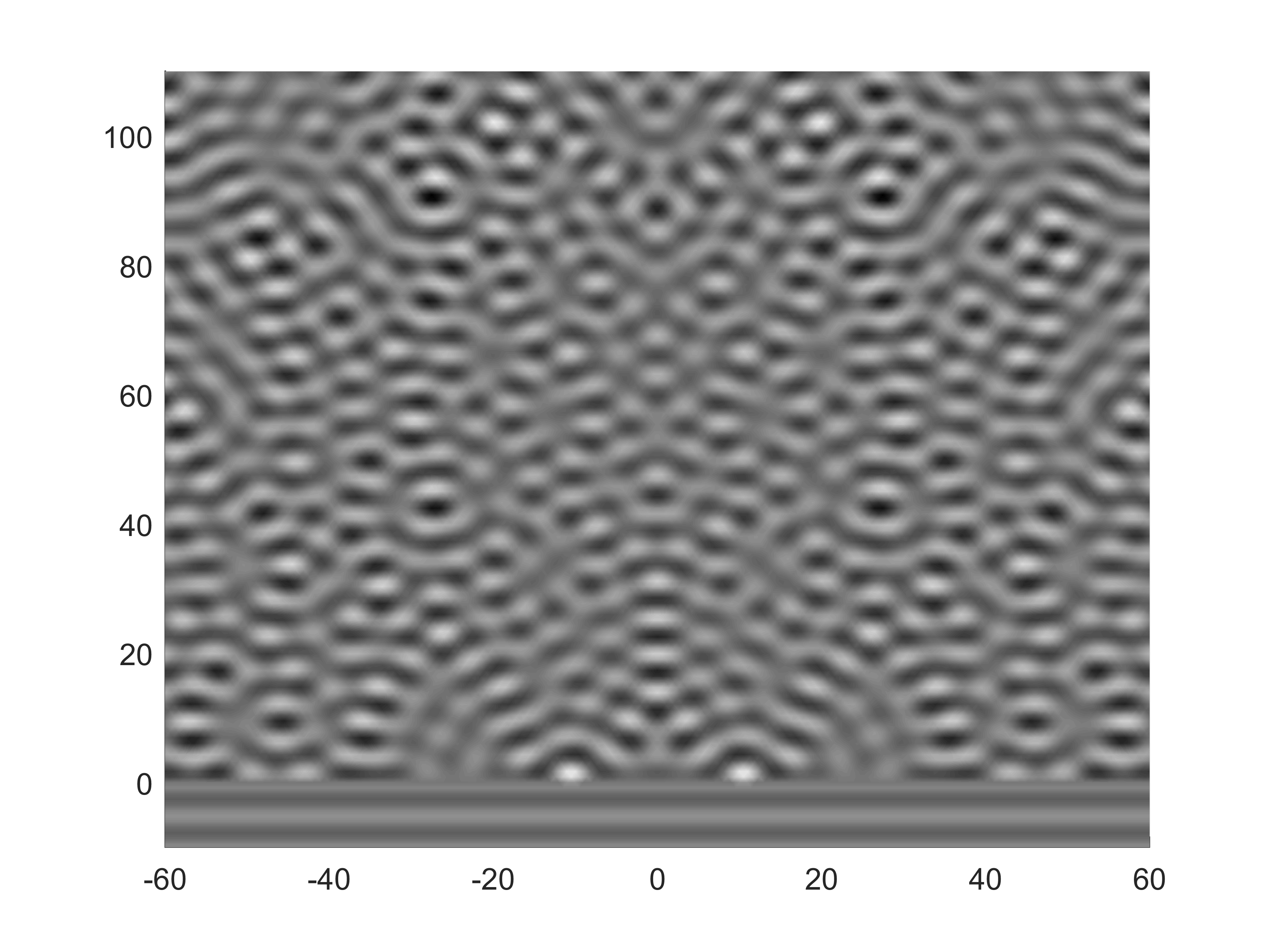}}
\begin{center}
\subfigure[The graph of $\mathrm{Re}\,u(\cdot,25\sqrt{3})$.]{\includegraphics[scale=0.6]{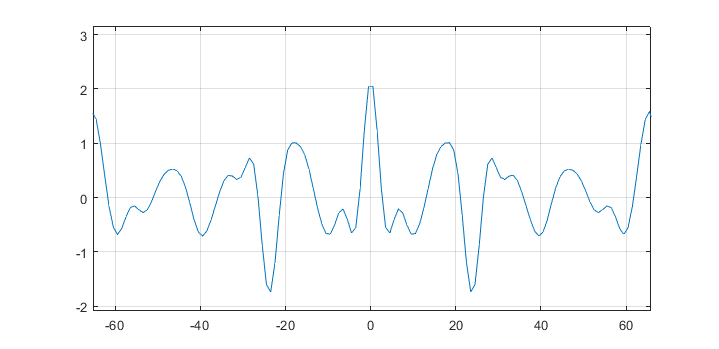}}
\end{center}
\caption{Plots are represented in original coordinates of the triangular lattice $\mathfrak{T}$. $k=\sqrt{2}$.}
\label{Fig_2}
\end{figure}

The results of numerical evaluations are plotted in Figure \ref{Fig_2} in original coordinates of $\mathfrak{T}$, where (a) shows the density plots of the Green's function $\mathrm{\Re e}\,\mathcal{G}$, (b) shows the density plots of $\cos(\iota x_2 \pi/3)$ on the lower half-plane and
$\mathrm{\Re e}\,u$ for the case $k=\sqrt{2}$ on the upper half-plane, and (c) presents the graph of $\mathrm{Re}\,u(\cdot,25\sqrt{3})$.  Some key features of the numerical solution can be immediately observed, namely, as expected, the symmetry of $\mathrm{\Re e}\,u$ and contributions from the wavefront that create a variable intensity.

\section{Discussion}

In this paper, we have constructed the discrete scattering theory for the two-dimensional discrete Helmholtz equation with a real wave number $k\in (0,3)\backslash\{2\sqrt{2}\}$ for the semi-infinite triangular lattices. The main objective was to prove the unique solvability result and derive an exact formula for the solution to a Dirichlet problem. For simplicity we restricted ourselves to compact boundary data. For non-compact boundary data we should require some extra conditions at infinity, cf. \cite{DK3}.

Similarly to the continuum theory, we used the notion of the radiating solution for the continuous Helmholtz equation
\begin{equation}\label{eq:Helmc}
\Delta u(x)+ \kappa^2 u(x)=0,\quad \kappa\in \mathbb{R}\backslash \{0\},\  x\in \mathbb{R}^2.
\end{equation}
Recall that a solution $u$ to the equation \eqref{eq:Helmc} is called radiating if it satisfies the Sommerfeld radiation condition
\begin{equation}\label{eq:zomm}
  \frac{\partial}{\partial \mid x \mid}u(x) - \iota \mid\kappa\mid u(x)=o(\mid x \mid ^{-\frac12})
\end{equation}
for $\mid x \mid\to \infty$ uniformly in all directions $x/\mid x \mid$. There are evident similarities and differences between conditions \eqref{eq:radcond}, \eqref{eq:radcond2} and \eqref{eq:zomm}.
For example, we see the same decay at infinity, but, in contrast to \eqref{eq:zomm}, the second condition of \eqref{eq:radcond} or \eqref{eq:radcond2} is anisotropic: the factor $e^{\iota \xi^*_j(\alpha,k)}$ is not a constant. Moreover, it is well known that the radiating solutions $u$ to the continuous Helmholtz equation automatically satisfy the Sommerfeld's finiteness condition, however this fact is not shown for the discrete case so far. Therefore, the first condition of \eqref{eq:radcond} or \eqref{eq:radcond2} is included in our definition.

In the present paper, the problems under consideration have an infinite boundary. Within the continuum framework it is well-known that, in general, when the surface is unbounded, we cannot neglect the contribution of that surface waves at infinity. In this case, the Sommerfeld radiation condition is no longer appropriate and a proper modification is needed. To the best of the authors' knowledge, different radiation conditions are provided only for the half-plane (and locally perturbed half-plane) problems, cf. \cite{CW, CW1, CWZ, CWZ1, DMN, DMN1}, and finding radiation conditions for arbitrary wedge-shaped regions remains open. In case of the square lattice we have proposed sufficient conditions for the given boundary data at infinity, which ensures to have an unique radiation solution to the corresponding problem, cf. \cite{DK3}. Here, in order to avoid further technical difficulties, we restricted ourselves with compactly supported data on the boundary.

Finally, let us note that comparing the results of problems in the continuous framework and results obtained in the continuum limit of discrete problems deserves the high interest of scientists, however, it is beyond the scope of this paper.

\section*{Acknowledgments}

This work was supported by Shota Rustaveli National Science Foundation of Georgia (SRNSFG) [FR-21-301]



\noindent

\section*{Appendix A Sparse matrices}\label{secA1}

The sparse matrices $\alpha_n(k)$, $\beta_n(k)$ and $\gamma_n(k)$ are defined as follows:
if $n=2p$  then  $\alpha_{2p}(k)$ is a $(p+1)\times p$ matrix such that
$\alpha_{2p}(k)\mid_{i,i}=1$, $i=\overline{1,p}$,
$\alpha_{2p}(k)\mid_{i,i-1}=1$, $i=\overline{2,p}$,
while $\alpha_{2p}(k)\mid_{p+1,p}=2$,
and all other matrix elements are zero. The $\beta_{2p}(k)$ is a $(p+1)\times (p+1)$ matrix
such that $\beta_{2p}(k)\mid_{i,i}=1$, $i=\overline{1,p}$,
$\beta_{2p}(k)\mid_{i,i+1}=1$, $i=\overline{2,p}$,
while $\beta_{2p}(k)\mid_{p+1,p+1}=\beta_{2p}(k)\mid_{1,2}=2$, and all other matrix elements are zero.
The $\gamma_{2p}(k)$ is a $(p+1)\times (p+1)$ matrix such that $\gamma_{2p}(k)\mid_{i,i}=6-k^2$, $i=\overline{1,p+1}$,
$\gamma_{2p}(k)\mid_{i,i+1}=\gamma_{2p}(k)\mid_{i,i-1}=-1$, $i=\overline{2,p}$, and $\gamma_{2p}(k)\mid_{1,2}=\gamma_{2p}(k)\mid_{p+1,p}=-2$.

If $n=2p+1$ then $\alpha_{2p+1}(k)$ is a $(p+1)\times (p+1)$ matrix such that
$\alpha_{2p+1}(k)\mid_{i,i}=1$, $i=\overline{1,p+1}$,
$\alpha_{2p+1}(k)\mid_{i,i-1}=1$, $i=\overline{2,p+1}$,
and all other matrix elements are zero. The $\beta_{2p+1}(k)$ is a
$(p+1)\times (p+2)$ matrix such that $\beta_{2p+1}(k)\mid_{i,i}=1$, $i=\overline{1,p+1}$,
$\beta_{2p+1}(k)\mid_{i,i+1}=1$, $i=\overline{2,p+1}$,
while $\beta_{2p+1}(k)\mid_{1,2}=2$, and all other matrix elements are zero.
The $\gamma_{2p}(k)$ is a $(p+1)\times (p+1)$ matrix such that $\gamma_{2p+1}(k)\mid_{i,i}=6-k^2$, $i=\overline{1,p}$,
$\gamma_{2p+1}(k)\mid_{p+1,p+1}=5-k^2$, while $\gamma_{2p+1}(k)\mid_{i,i+1}=-1$, $i=\overline{2,p}$, $\gamma_{2p+1}(k)\mid_{i,i-1}=-1$, $i=\overline{2,p+1}$, and $\gamma_{2p+1}(k)\mid_{1,2}=\gamma_{2p+1}(k)\mid_{p+1,p}=-2$. Finally, $\gamma_1(k)$ is a $1\times 1$ matrixs with an element $4-k^2$.

\end{document}